\documentclass[conference,10 pt,onecolumn]{IEEEtran}       

\usepackage[margin=0.701in]{geometry}                      

\usepackage{cite}
\usepackage{ifthen}
\usepackage[cmex10]{amsmath}
\usepackage{amsfonts,amssymb,amsthm,bbm,verbatim}

\usepackage{array}

\usepackage[pdftex]{graphicx}

\usepackage[utf8]{inputenc}
\usepackage[T1]{fontenc}
\usepackage{url}            
\usepackage[usenames,dvipsnames]{color}
\usepackage{booktabs}       
\usepackage[nice]{nicefrac}
\usepackage{mathtools}
\usepackage{siunitx} 
\usepackage{threeparttable}
\usepackage{dsfont}
\usepackage{placeins}
\usepackage{xcolor}
\usepackage{enumitem}

\makeatletter
\newcommand\footnoteref[1]{\protected@xdef\@thefnmark{\ref{#1}}\@footnotemark}
\makeatother

\newtheorem{lemma}{Lemma} 
\newtheorem{theorem}[lemma]{Theorem}
\newtheorem{corollary}[lemma]{Corollary}
\newtheorem{conjecture}[lemma]{Conjecture}

\newtheorem{proposition}[lemma]{Proposition} 
\newtheorem{definition}[lemma]{Definition}

\newtheorem{remark}[lemma]{Remark}

\DeclareMathOperator*{\argmin}{arg\,min}

\newcommand{\Rb}{\mathbb{R}}
\newcommand{\Kcal}{\mathcal{K}}

\newcommand{\Scal}{\mathcal{S}}

\newcommand{\Ycal}{\mathcal{Y}}
\newcommand{\Zcal}{\mathcal{Z}}

\newcommand{\SK}[3]{S_{\rightarrow}\!\left({#1};{#2}\!\left|{#3} \right. \right)} 
\newcommand{\SKK}[3]{S_{\leftrightarrow}\!\left({#1};{#2}\!\left|{#3} \right. \right)} 

\definecolor{lnkcolr}{rgb}{0.89, 0.08, 0.17} 
\definecolor{urlcolr}{rgb}{0.0, 0.35, 0.26}

\IEEEoverridecommandlockouts          

\usepackage{hyperref}
\hypersetup{colorlinks = true, linkcolor = lnkcolr, citecolor = blue, urlcolor=urlcolr}

\begin{document}
\title{Unique Information and Secret Key Decompositions}

\author{Johannes Rauh$^{\ast}$, Pradeep Kr. Banerjee$^{\ast}$, Eckehard Olbrich and J\"urgen Jost\\
	\IEEEauthorblockA{Max Planck Institute for Mathematics in the Sciences, Leipzig, Germany\\
	Email: \tt\small\{jrauh,pradeep,olbrich,jjost\}@mis.mpg.de}
	\thanks{$^{\ast}$The first two authors contributed equally to this work.}
	}

\maketitle

\thispagestyle{plain}          
\pagestyle{plain}              

\begin{abstract}
The \emph{unique information} ($UI$) is an information measure that quantifies a deviation from the Blackwell order. We have recently shown that this quantity is an upper bound on the \emph{one-way secret key rate}. In this paper, we prove a triangle inequality for the~$UI$, which implies that the $UI$ is never greater than one of the best known upper bounds on the \emph{two-way secret key rate}. We conjecture that the~$UI$ lower bounds the two-way rate and discuss implications of the conjecture.
\end{abstract}

{\hypersetup{linkcolor=black}
 \tableofcontents
}

\section{Introduction}
We consider the well-known \emph{source model} for secret key agreement \cite{maurer1993,maurerintrinsic}: Alice, Bob and an adversary Eve observe~$n$ i.i.d. copies of random variables~$S$, $Y$ and~$Z$ resp., where~$(S,Y,Z)$ is distributed according to some joint distribution known to all parties. Alice and Bob wish to agree on a common secret key, by publicly communicating messages over an authenticated and noiseless channel transparent to Eve. 

A \emph{two-way} public communication protocol proceeds in rounds, where Alice and Bob exchange messages in alternating order, with Alice sending messages in the odd rounds and Bob in the even rounds. Each message is a function of the sender's observation and all the messages exchanged so far. At the end of the protocol, Alice (resp., Bob) computes a key~$K$ (resp.,~$K'$) as a function of~$S^n$ (resp.,~$Y^n$) and~$C$, the set of all exchanged messages.
\begin{definition}[$\!\!\!$\cite{maurer1993}]
The \emph{two-way secret key rate}, denoted $\SKK{S}{Y}{Z}$, is the maximum rate~$R$ such that for every~$\epsilon > 0$, and for sufficiently large~$n$, there exists a public communication protocol such that~$K$ and~$K'$ (ranging over some common set~$\Kcal$) agree with probability at least~$1-\epsilon$, satisfying
\begin{align}\label{eq:SKdefinition}
\tfrac{1}{n}H(K) > \tfrac{1}{n}\log|\Kcal|-\epsilon, \quad  \tfrac{1}{n}I(K;C,Z^n)\le \epsilon, 
\end{align}
and achieving~$\tfrac{1}{n}H(K)\ge R-\epsilon$.
\end{definition}
\eqref{eq:SKdefinition} ensures that the key is almost uniformly distributed and that the \emph{rate} at which Eve learns information about the key is negligibly small. 
A still stronger definition requires that Eve's \emph{total} information about the key is negligibly small. 
For the source model, both these definitions give the same secret key rates~\cite{maurerstrong}. We refer~\cite[Section 17.3]{ckbook} for a review.

The protocol is \emph{one-way} if there is only one round of communication from Alice to Bob. The corresponding key rate is called the \emph{one-way secret key rate} $\SK{S}{Y}{Z}$. 
The one-way secret key rate is a lower bound on the two-way secret key rate. The former 
can be expressed as an optimization problem over Markov kernels of bounded size~\cite{csiszar1978broadcast,ahlswede1993}. 
In contrast, no algorithm to compute the two-way key rate is known, and its value is known only for a handful of distributions~\cite{ahlswede1993,gohari1,goharicomments,chitambarSK}.
Computing the two-way rate is a fundamental and open area of inquiry in information-theoretic cryptography.

The state-of-the-art upper bounds on the two-way key rate rely on the following key observation~\cite{gohari1,gohari3}:
Let $s=\SKK{S}{Y}{Z}$. Imagine a fourth party Charlie who observes i.i.d. copies of a correlated random variable~$Z'$. If we decompose~$s$ into two parts: a part~$s_1$ which Charlie does not know, and a part~$s_2=s-s_1$ which Charlie knows about the secret key shared between $S$ and $Y$ w.r.t.~$Z$, then $s_1$ is at most~$\SKK{S}{Y}{Z'}$, while~$s_2$ is at most~$\SK{SY}{Z'}{Z}$. 
Thus, for any $(S,Y,Z,Z')$, the secret key rate satisfies the following property~\cite[Theorem 4]{gohari1}.
\begin{align}\label{eq:SKdecomposition}
\SKK{S}{Y}{Z}&\leq \SKK{S}{Y}{Z'}+\SK{SY}{Z'}{Z}.
\end{align}

For any~$(S,Y,Z,Z')\sim P$, if the induced channel~$P_{Z|SY}$ dominates the channel~$P_{Z'|SY}$ in the \emph{less noisy} sense~\cite{KoernerMarton75:Comparison_of_noisy_channels}, then the second term~$\SK{SY}{Z'}{Z}$ vanishes. 
One can thus interpret the second term in~\eqref{eq:SKdecomposition} as quantifying a deviation from the less noisy order when we replace~$P_{Z|SY}$ with~$P_{Z'|SY}$~\cite{gohariachieving}.

The secret key rates are similar in spirit to an information theoretic quantity $UI$, called \emph{unique information}, first proposed in~\cite{e16042161}. The value $UI(S;Y\backslash Z)$ is interpreted as the \emph{information about $S$ known to $Y$, but unknown to~$Z$}. 
The definition of $UI$ is motivated by the idea that unique information should be \emph{useful}.
In~\cite{e16042161} this is formalized in terms of decision problems: whenever $UI(S;Y\backslash Z)>0$, there is a decision problem in which it is better to know $Y$ than to know~$Z$. 
A second ingredient is the goal to not only measure some aspect of information, but also to define an information decomposition, in the sense of~\cite{WilliamsBeer}, that is,
\begin{equation}
  \label{eq:SI}
  SI(S;Y,Z) = I(S;Y) - UI(S;Y\backslash Z)
\end{equation}
is nonnegative and can be interpreted as the information about $S$ \emph{shared} between~$Y$ and~$Z$, and
\begin{equation}
  \label{eq:CI}
  CI(S;Y,Z) = I(S;Y|Z) - UI(S;Y\backslash Z)
\end{equation}
is nonnegative and can be interpreted as \emph{synergistic} (or \emph{complementary}) information about~$S$. One can thus interpret the unique information as either the mutual information without the shared information, or as the conditional mutual information without the synergistic information.

The key rates can be described in a similar manner as \emph{information common to~$S$ and~$Y$ that is unique w.r.t.~$Z$}.
Also it is clear by definition in which sense positive values of the key rates are useful. Thus it is natural to ask how the two concepts are related. By studying this question we hope to further both the understanding of the secret key rates and the understanding of information decompositions: in fact, the function $UI$ has been criticized amongst other things for vanishing too often. For example, $UI(S;Y\backslash Z)=0$ whenever the marginals $(S,Y)$ and $(S,Z)$ are identically distributed. The two-way key rate can still be positive in such a situation (see e.g., \cite{UIdefAllerton,goharicomments}).

In~\cite{UIdefAllerton}, we have recently shown that $UI$ is an upper bound on the one-way secret key rate. 
We have also shown that neither the one-way nor the two-way key rate directly fits into the information decomposition framework, as it violates a so-called consistency condition, but we presented a simple construction to enforce the consistency condition and nevertheless derive an information decomposition.

In this paper, we prove a triangle inequality for the~$UI$ which implies the following property that resembles~\eqref{eq:SKdecomposition}:
	For any $(S,Y,Z,Z')$, 
	\begin{align}\label{eq:UISKdecomposition}
	UI(S;Y\backslash Z) \le UI(S;Y\backslash Z')+UI(SY;Z' \backslash Z).
	\end{align}	
From \eqref{eq:UISKdecomposition} we conclude that $UI\le B_{1}$, where $B_{1}$ is one of the best known upper bounds on $S_{\leftrightarrow}$. We conjecture that the~$UI$ lower bounds the two-way key rate and discuss implications of the conjecture.

\section{The Unique Information and its Properties}
For some finite state spaces $\Scal,\Ycal,\Zcal$, let $\mathbb{P}_{\Scal\times\Ycal\times\Zcal}$ be the set of all joint distributions of $(S,Y,Z)$. 
Given $P \in \mathbb{P}_{\Scal\times\Ycal\times\Zcal}$, let 
\begin{equation}
\Delta_{P(S,Y,Z)} := \big\{Q \in \mathbb{P}_{\Scal\times\Ycal\times\Zcal}\colon  Q_{SY}(s,y)=P_{SY}(s,y),\;
Q_{SZ}(s,z)=P_{SZ}(s,z)\big\}
\label{eq:delP}
\end{equation}
be the set of joint distributions of $(S,Y,Z)$ that have the same marginals on $(S,Y)$ and $(S,Z)$ as $P$. 
For brevity, we sometimes write~$\Delta_{P(S,Y,Z)}\equiv \Delta_P$.
\cite{e16042161} define the unique information that $Y$ conveys about $S$ w.r.t. $Z$ as 
\begin{align}
UI(S;Y\backslash Z) := \min_{Q \in \Delta_{P(S,Y,Z)}} I_Q(S;Y|Z),
\label{subeq:UIy}
\end{align}
where the subscript~$Q$ in~$I_Q$ denotes the joint distribution on which the mutual information $I$ is computed.
Since~$\Delta_{P}$ is compact and $I_{Q}$ is continuous in~$Q$, the minimum exists. 
$\Delta_{P}$ is a convex polytope of dimension~$|\Scal|(|\Ycal|-1)(|\Zcal|-1)$,
and the optimization problem~\eqref{subeq:UIy} is a convex program~\cite{e16042161}, actually a convex cone program~\cite{DOT2017bivariate}.
An algorithm to compute the $UI$ has been proposed in~\cite{CUIfullver}\footnote{Link to source code is available at \url{https://github.com/infodeco/computeUI}.}.

The function~$UI$ satisfies the following consistency condition, which implies that $SI$ and $CI$ (defined in~\eqref{eq:SI} and~\eqref{eq:CI}) are symmetric in~$Y,Z$~\cite{e16042161}.
\begin{enumerate}[label=\textbf{P.\arabic*}]
	\item \label{lem:consistency} (\textit{Consistency condition}).
	\begin{align} \label{eq:consistency}
	I(S;Y)+UI(S;Z \backslash Y)=I(S;Z)+UI(S;Y \backslash Z).
	\end{align}
\end{enumerate}
$UI$ also satisfies the following intuitive property. 
\begin{enumerate}[resume,label=\textbf{P.\arabic*}]
	\item \label{lem:BP} (\textit{Blackwell property}). For~$(S,Y,Z)\sim P$, write~$Z\succeq_{S} Y$ if there exists a random variable~$Y'$ such that $S-Z-Y'$ is a Markov chain and $P_{SY'}=P_{SY}$. Then $UI(S;Y\backslash Z)$ vanishes if and only if~$Z\succeq_{S} Y$~\cite[Lemma 6]{e16042161}.
\end{enumerate}
Blackwell's theorem~\cite{Blackwell1953,BlackwellISIT} implies that this property is equivalent to the fact that decision problems can be solved using~$Z$ at least as well as with~$Y$.
We call~$\succeq_{S}$ the~\emph{Blackwell order} (also called the degradation order). The~$UI$ then quantifies a deviation from the Blackwell order.

\subsection{Monotonicity properties of the unique information} 
In this section, we review basic properties that the function~$UI$ shares with the two-way secret key rate.  
We first note the following trivial bounds~\cite{e16042161}. 
\begin{align} 
	I(S;Y)-I(S;Z) \le UI(S;Y\backslash Z) \le \min\{I(S;Y),I(S;Y|Z)\}.
\end{align}
These bounds match the trivial bounds on the two-way secret key rate~\cite{maurerintrinsic} (note that $\SKK SYZ$ is symmetric under permutations of $S,Y$, while $UI(S;Y\backslash Z)$ is not):
\begin{align} \label{eq:trivialboundSK} 
	\max\{I(S;Y)-I(S;Z),I(Y;S)-I(Y;Z)\}\le \SKK{S}{Y}{Z}\le \min\{I(S;Y),I(S;Y|Z)\}.
\end{align}
In a secret key agreement task, if either Eve has less information about~$S$ than Bob or, by symmetry, less information about~$Y$ than Alice, then Alice and Bob can exploit this difference to extract a secret key.

In~\cite{UIdefAllerton}, we proved the following properties of the~$UI$.
\begin{enumerate}[resume,label=\textbf{P.\arabic*}]
	\item \label{lem:LOAliceBob} (\textit{Monotonicity under local operations of Alice and Bob}). For all~$(S,S',Y,Z)$ such that~$YZ$--$S$--$S'$ is a Markov chain, $UI(S;Y\backslash Z) \ge UI(S';Y\backslash Z)$ (and likewise for local operations on~$Y$).
	\item \label{lem:PC} (\textit{Monotonicity under public communication by Alice}). For all~$(S,Y,Z)$ and functions~$f$ over the support of $S$, 
	\begin{align*}   
	   UI((S,f(S));(Y,f(S)) \backslash (Z,f(S))) \leq UI(S;Y \backslash Z). 
	\end{align*}
	\item \label{lem:Normalization} (\textit{Normalization}). For a perfect secret bit $P_{SSZ}(0,0|z)=P_{SSZ}(1,1|z)=\tfrac{1}{2}$, $UI(S;S\backslash Z)=1$. 
	\item \label{lem:AD} (\textit{Additivity on tensor products}). For~$n$ i.i.d. copies of $(S,Y,Z)\sim P$, $UI(S^n;Y^n\backslash Z^n) = n\cdot UI(S;Y\backslash Z)$.
    \item \label{thm:AC} (\textit{Asymptotic continuity}). For any~$P,P'\in \mathbb{P}_{\Scal\times\Ycal\times\Zcal}$, and~$\epsilon\in [0,1]$, if~$\|P-P'\|_1= \epsilon$, then
      \begin{align*} 
        {UI}_{P'}(S;Y\backslash Z) - {UI}_{P}(S;Y\backslash Z) \le \zeta(\epsilon) + 5 \epsilon \log\min\{|\Scal|, |\Ycal|\}
      \end{align*}
for some bounded, continuous function $\zeta:[0,1]\to \Rb_+$ such that $\zeta(0)=0$.
\end{enumerate}

\subsection{A triangle inequality for the unique information}  

In this section, we prove the following triangle inequality. 
\begin{proposition}\label{prop:triangleUI}
	For any $(S,Y,Z,Z')$, 
	$
	UI(S;Y\backslash Z) \le UI(S;Y\backslash Z')+UI(S;Z' \backslash Z).
	$
\end{proposition}

To prove Proposition~\ref{prop:triangleUI}, we need the following monotonicity property of the function~$UI$ that is proved in the appendix. 
\begin{enumerate}[resume,label=\textbf{P.\arabic*}]
	\item \label{lem:LOEve} (\textit{Monotonicity under local operations of Eve}). For all $(S,Y,Z,Z')$ such that~$SY$--$Z$--$Z'$ is a Markov chain, $UI(S;Y\backslash Z) \le UI(S;Y\backslash Z')$.
\end{enumerate}
For the special case when~$Z'$ is a deterministic function of~$Z$, Property~\ref{lem:LOEve} was shown in~\cite{ISIT_RBOJ14}.

One can gain an intuitive understanding of Proposition~\eqref{prop:triangleUI} by iterating the basic information decomposition idea as follows.  In the presence of a fourth variable $Z'$, we would like to decompose $u := UI(S;Y\backslash Z)$ into two parts: a part $u_1$, which $Z'$ also knows, and the remainder $u_2 = u- u_1$, which $Z'$ does not know. 
Clearly, $u_1$ should be upper bounded by $UI(S;Z'\backslash Z)$ since $Z'$ alone knows what $Z'$ and $Y$ share. Furthermore, $u_2 \le UI(S;Y\backslash Z')$ since what neither $Z$ nor $Z'$ knows is less than what $Z'$ does not know. 
In total this gives a heuristic argument why the triangle inequality should hold. 
\begin{proof} 
	Let~$(S,Y,Z,Z')\sim P$. 
	We use the following notation: For~$A,B\subseteq\{Y,Z,Z'\}$, $\Delta_{P(S,A,B)}$ is the set of all joint distributions of~$(S,A,B)$ that have the same marginals on the pairs~$(S,A)$ and~$(S,B)$ as~$P$.
	
	Let $Q^{\ast} \in \argmin_{Q\in\Delta_{P(S,Z',Z)}}I(S;Z'|Z)$. Extend $Q^{\ast}$ to a distribution of $S,Y,Z',Z$ via
	\begin{align*} 
      Q^{\ast}(s,y,z',z) = Q^{\ast}(s,z',z) P(y|s,z')
      \text{ if~$P(s,z')>0$,}
	\end{align*}
	and $Q^{\ast}(s,y,z',z)=0$ otherwise.
	Then $Q^{\ast}(S,Y,Z') = P(S,Y,Z')$ and $Q^{\ast}(S,Y,Z) \in \Delta_{P(S,Y,Z)}$.
	Thus,
	\begin{align*} 
              \min_{Q\in\Delta_{Q^{\ast}(S,Y,Z'Z)}} I(S;YZ'|Z) 
	&\stackrel{(a)}{=} \min_{Q\in\Delta_{Q^{\ast}(S,Y,Z'Z)}} I(S;Z'|Z) + I(S;Y|Z'Z)
	\\ &\stackrel{(b)}{=} UI(S;Z'\backslash Z)
	+ \min_{Q\in\Delta_{Q^{\ast}(S,Y,Z'Z)}} I_{Q}(S;Y|Z'Z) 
	\\ &\stackrel{(c)}{=} UI(S;Z'\backslash Z) + UI_{Q^{\ast}}(S;Y\backslash Z'Z)
	\\ &\stackrel{(d)}{\le} UI(S;Z'\backslash Z) + UI_{Q^{\ast}}(S;Y\backslash Z')
               = UI(S;Z'\backslash Z) + UI(S;Y\backslash Z').
	\end{align*}
	(a) follows from the chain rule of mutual information. (b) follows since the $(S,Z,Z')$-marginal is fixed in $\Delta_{Q^{\ast}(S,Y,Z'Z)}$ and by definition of~$Q^{\ast}$, $I_{Q^{\ast}}(S;Z'|Z)=UI(S;Z'\backslash Z)$. (c) follows because the second minimization in~(b) defines $UI_{Q^{\ast}}(S;Y\backslash Z'Z)$. Finally, (d) follows from Property~\ref{lem:LOEve}.
	
	Let $Q^{+} \in \argmin_{Q\in\Delta_{Q^{\ast}(S,Y,Z'Z)}} I(S;YZ'|Z)$.
	Then
	\begin{equation*}
	Q^{+}(S,Y,Z) \in\Delta_{Q^{\ast}(S,Y,Z)} = \Delta_{P(S,Y,Z)}.  
	\end{equation*}
	Therefore,
	\begin{align*} 
	\min_{Q\in\Delta_{Q^{\ast}(S,YZ',Z)}} &I(S;YZ'|Z) = I_{Q^{+}}(S;YZ'|Z)
          \ge UI_{Q^{+}}(S;YZ'\backslash Z)
          \ge UI_{Q^{+}}(S;Y\backslash Z)= UI(S;Y\backslash Z),
	\end{align*}
	where in the last step we have used Property~\ref{lem:LOAliceBob} and the fact that~$Q^{+}(S,Y,Z)\in \Delta_{P(S,Y,Z)}$. 
	This completes the proof.
\end{proof}

From Proposition~\ref{prop:triangleUI} and Property~\ref{lem:LOAliceBob} we conclude:
\begin{corollary}\label{corr:UISKdecomposition}
	For any $(S,Y,Z,Z')$,  
	$
	UI(S;Y\backslash Z) \le UI(S;Y\backslash Z')+UI(SY;Z' \backslash Z).
	$
\end{corollary}
We can interpret Corollary~\ref{corr:UISKdecomposition} like inequality~\eqref{eq:SKdecomposition}: Given $(S,Y,Z,Z')\sim P$, if the induced channel $P_{Z|SY}$ dominates the channel $P_{Z'|SY}$ in the Blackwell sense (see Property~\ref{lem:BP}).
then the second term~$UI(SY;Z' \backslash Z)$ vanishes. One can interpret $UI(SY;Z' \backslash Z)$ as quantifying a deviation from the Blackwell order when we replace~$P_{Z|SY}$ with~$P_{Z'|SY}$.

\section{Bounds on Secret Key Rates}

\subsection{An upper bound on the one-way secret key rate}
\label{sec:one-waySK}

$S_{\to}$ admits the following characterization.  
\begin{theorem}[$\!\!${\cite[Theorem~1]{ahlswede1993}}]\label{thm:skaRate}
	The \emph{one-way secret key rate} $\SK{S}{Y}{Z}$ for the source model is
	\begin{align*}
		\SK{S}{Y}{Z}=\max\limits_{P_{UV|SYZ}} & I(U;Y|V)-I(U;Z|V) 
	\end{align*}
	for random variables~$U$, $V$ of bounded cardinalities $|\mathcal{U}|\le|\Scal|^2$ and $|\mathcal{V}|\le|\Scal|$, such that $V$--$U$--$S$--$YZ$ is a Markov chain.
\end{theorem}
The bounds on the cardinalities imply that the optimization domain is a set of stochastic matrices of finite size, which makes it possible to turn this theorem into an algorithm to compute $S_{\rightarrow}$.

Like the~$UI$, $\SK{S}{Y}{Z}$ depends only on the marginal distributions of the pairs~$(S,Y)$ and~$(S,Z)$~\cite{ahlswede1993}. 
Using Properties~\ref{lem:LOAliceBob} -- \ref{thm:AC} and results about protocol monotones \cite{maurerunbreakable,gohari1,christandl,gohari3}, one can show the following:
\begin{theorem}[$\!\!${\cite[Theorem 37]{UIdefAllerton}}] \label{thm:uppbound}
	$UI(S;Y\backslash Z)$ is an upper bound for the one-way secret key rate~$\SK{S}{Y}{Z}$.
\end{theorem}

\subsection{Known upper bounds on the two-way secret key rate}
\label{subsec:two-waySK}
As noted in~\eqref{eq:trivialboundSK}, a trivial upper bound on $\SKK{S}{Y}{Z}$ is $\min\{I(S;Y), I(S;Y|Z)\}$~\cite{maurer1993}. 
An improved upper bound is given by the \emph{intrinsic information}~\cite{maurerintrinsic}.
\begin{align}
	\SKK{S}{Y}{Z}\leq I(S;Y\!\!\downarrow \!Z):= \min_{P_{Z'|Z}}I(S;Y|Z'),\label{eq:B0upperbound}
\end{align}
where~$Z'$ may be assumed to have a range of size at most~$|\Zcal|$~\cite{christandlICMI}.

\cite{RennerW03} noted that the intrinsic information exhibits a property called ``locking'', i.e., it can drop by an arbitrarily large amount on giving away a bit of information to Eve. 
In contrast, the two-way rate satisfies
\begin{align}\label{eq:SKlockingproperty}
	\SKK{S}{Y}{ZU}\ge \SKK{S}{Y}{Z}-H(U)
\end{align}
for jointly distributed random variables~$(S,Y,Z,U)$~\cite[Theorem 3]{RennerW03}, and the conditional mutual information satisfies an analogous property.
The same is true for the~$UI$:
\begin{enumerate}[resume,label=\textbf{P.\arabic*}]
\item \label{lem:UISKlock} (\textit{$UI$ does not lock}). For jointly distributed random variables $(S,Y,Z,U)$, 
  \begin{align}
	UI(S;Y\backslash ZU) \ge UI(S;Y\backslash Z)-H(U). \label{eq:UISKlocking property}
  \end{align}
\end{enumerate}
The proof of Property~\ref{lem:UISKlock} is in the appendix.

\cite{RennerW03} proposed an improved upper bound called the \emph{reduced intrinsic information}, which does not exhibit locking:
\begin{equation*}
  I(S;Y\!\!\downarrow\downarrow \!Z) := \inf_{P_{U|SYZ}} I(S;Y\!\!\downarrow \!ZU)+H(U).
\end{equation*}
Property~\ref{lem:UISKlock} implies that $UI(S;Y\setminus Z) \le I(S;Y\!\!\downarrow\downarrow \!Z)$; a fact that will be generalized later in Thereom~\ref{thm:fullchain}.

The tightest known upper bound on the two-way rate is \cite{gohari1} 
\begin{equation}
  B_{2}(S;Y|Z) := \inf_{p_{Z'|SYZ}}I(S;Y|Z') + \SK{SY}{Z'}{Z}. \label{eq:B2upperbound}
\end{equation}
Unfortunately, $B_2$ cannot be computed explicitly, as no bound on the size of~$Z'$ is known. 

A slightly weaker but computable upper bound is given by the \emph{minimum intrinsic information}~\cite{gohari1}.
\begin{align}
	B_1(S;Y|Z):=\min_{P_{Z'|SYZ}} I(S;Y|Z')+I(SY;Z'|Z),\label{eq:B1upperbound}
\end{align} 
where~$|\mathcal{Z}'|\le |\Scal||\Ycal||\Zcal|$. 

\subsection{Unique information based bounds on the two-way rate and a conjecture} 
\label{subsec:UItwo-waySK}
We now investigate some properties of the function~$UI$ in relation to upper bounds on the two-way rate. 
We first list the following known chain of bounds on the two-way rate. 
\begin{align} \label{eq:sandwichboundSK} 
\!\!\SK{S}{Y}{Z}\le \SKK{S}{Y}{Z} \le B_2(S;Y|Z)\le B_1(S;Y|Z) 
\le I(S;Y\!\!\downarrow\downarrow \!Z) \le I(S;Y\!\!\downarrow \!Z)\le I(S;Y|Z).
\end{align}

Corollary~\ref{corr:UISKdecomposition} implies the following result.
\begin{proposition}\label{prop:UIB1}
  $UI(S;Y\backslash Z) \le B_1(S;Y|Z).$
\end{proposition}

From Theorem~\ref{thm:uppbound} and Proposition~\ref{prop:UIB1}, we have the following chain of inequalities relating the bounds on the two-way rate.
\begin{theorem}\label{thm:fullchain}
  $
	\SK{S}{Y}{Z} \le UI(S;Y\backslash Z) \le B_1(S;Y|Z)
    \le I(S;Y\!\!\downarrow\downarrow \!Z) \le I(S;Y\!\!\downarrow \!Z)\le I(S;Y|Z).
  $
\end{theorem}
Given~$(S,Y,Z)\sim P$, let 
\begin{align}\label{eq:Qstar}
Q^{*} \in \argmin_{Q\in\Delta_{P(S,Y,Z)}} I_{Q}(S;Y|Z).
\end{align}
The distribution $Q^{\ast}$ is called a \emph{minimum synergy} distribution, as \mbox{$CI(S;Y,Z) = 0$} if and only if $P=Q^{\ast}$.
By definition, $I_{Q^{\ast}}(S;Y|Z)=UI(S;Y\backslash Z)$.
An immediate consequence of Theorem~\ref{thm:fullchain} is the following: 
Choosing~$P=Q^{\ast}$, all known upper bounds on the two-way rate collapse to the~$UI$ and the conditional mutual information, respectively.

Examples are known which show that~$UI$ is not an upper bound on~$S_{\leftrightarrow}$ (see e.g., \cite[Example 41]{UIdefAllerton}, \cite[Appendix]{goharicomments}).
We make the following conjecture.
\begin{conjecture}\label{conj:UISK}
  $UI(S;Y\backslash Z) \le \SKK{S}{Y}{Z}$.
\end{conjecture}
Let us briefly mention why we believe that this conjecture is true. Firstly, while the function~$UI(S;Y\backslash Z)$ depends only on the marginals of the pairs~$(S,Y)$ and~$(S,Z)$, the same is not true for~$\SKK{S}{Y}{Z}$ which depends on the full joint distribution of~$(S,Y,Z)$. Secondly, unlike~$\SKK{S}{Y}{Z}$ which is symmetrical in~$S$ and~$Y$, the function~$UI$ is asymmetric in all three variables. This asymmetry is manifest, for instance, when we note that~$UI(S;Y\backslash Z)$ is not monotone under public communication by Bob. 

\begin{remark}[Sandwich bound on~$\SKK{S}{Y}{Z}$]\label{rem:UI-2SK}
	 If Conjecture~\ref{conj:UISK} is true, then 
	 \begin{align}\label{eq:SKKsandwichbound}
	 UI(S;Y\backslash Z)= I_{Q^{*}}(S;Y|Z)\le \SKK{S}{Y}{Z} \le I_P(S;Y|Z).
	 \end{align}
	 \eqref{eq:SKKsandwichbound} implies that 
     the set of all $Q^{*}$ as in~\eqref{eq:Qstar} is a set of distributions for which the $UI$ equals the two-way rate. 
\end{remark}
A related work~\cite{chitambarSK} gives necessary conditions for when the two-way rate equals the conditional mutual information.

\begin{definition} \label{def:newUIbounds}
	Define the following functions on~$\mathbb{P}_{\Scal\times\Ycal\times\Zcal}$.
	\begin{align*}
      B_{sUI}(S;Y|Z)&:=\inf_{P_{Z'|SYZ}} UI(S;Y\backslash Z')+UI(SY;Z' \backslash Z). 
      \\
	B_{gUI}(S;Y|Z)&:=\inf_{P_{Z'|SYZ}} I(S;Y|Z')+UI(SY;Z' \backslash Z). 
	\end{align*}
\end{definition}
As the following proposition shows, $B_{gUI}(S;Y|Z)$ is a new upper bound on the two-way rate which is juxtaposed between the two best known bounds~$B_2$ and~$B_1$.
\begin{proposition} \label{prop:newUIbounds}	
	\begin{align}
	B_{sUI}(S;Y|Z) & = UI(S;Y\backslash Z) \le B_{gUI}(S;Y|Z)\label{eq:sgUIbound}\\
	B_2(S;Y|Z) &\le B_{gUI}(S;Y|Z) \le B_1(S;Y|Z)\label{eq:newUIbound}
	\end{align}
\end{proposition}
\begin{proof} 
	The left equality in~\eqref{eq:sgUIbound} follows from Corollary~\ref{corr:UISKdecomposition} and
	\begin{align*} 
	B_{sUI}(S;Y|Z)=\inf_{P_{Z'|SYZ}} UI(S;Y\backslash Z')+UI(SY;Z' \backslash Z)
	\le \inf_{P_{Z'|Z}:SY-Z-Z'} UI(S;Y\backslash Z') = UI(S;Y\backslash Z),
	\end{align*}
	where the last equality uses Property~\ref{lem:LOEve}.
	The right inequality in~\eqref{eq:sgUIbound} follows from Corollary~\ref{corr:UISKdecomposition} and from $UI(S;Y\backslash Z')\le I(S;Y|Z')$.
	
	Statement~\eqref{eq:newUIbound} follows from Theorem~\ref{thm:uppbound} by noting that $\SK{SY}{Z'}{Z}\le UI(SY;Z'\backslash Z) \le I(SY;Z'|Z)$.
\end{proof}

\section{Conclusion} 
We showed a triangle inequality for the unique information which implies that the $UI$ is never greater than one of the best known upper bounds on the two-way secret key rate. 
We conjecture that the~$UI$ is indeed a lower bound on the two-way rate.
Assuming that the conjecture is true, we characterized a set of distributions for which the two-way rate equals the conditional mutual information and the~$UI$.
This provides an operational characterization of the~$UI$.

\appendix
\label{app:proofs_UISK}
\begin{proof}[Proof of Property~\ref{lem:LOEve}]
	Let $(S,Y,Z)\sim P$ and $(S,Y,Z,Z')\sim P'$.  By definition, $P$ is a marginal of~$P'$. Let~$Q\in\Delta_{P(S,Y,Z)}$, and let
	$Q'(s,y,z,z') = Q(s,y,z) P'(z'|s,z)$ if~$P(s,z)>0$ and~$Q'(s,y,z,z')=0$ otherwise. Then~$Q'\in\Delta_{P'(S,Y,ZZ')}$. Moreover, $Q$ is the $(S,Y,Z)$-marginal of~$Q'$, and $Y$--$SZ$--$Z'$ is a Markov chain w.r.t.~$Q'$ by construction. Therefore,
	\begin{align*} 
	I_{Q'}(S;Y|ZZ') &= I_{Q'}(SZ';Y|Z) - I_{Q'}(Z';Y|Z)
	\\ 
	&\le I_{Q'}(SZ';Y|Z)
	= I_{Q'}(S;Y|Z) + I_{Q'}(Z';Y|SZ)
               = I_{Q'}(S;Y|Z) = I_{Q}(S;Y|Z).
	\end{align*}
	Taking the minimum over~$Q\in\Delta_{P(S,Y,Z)}$, we conclude that
	\begin{align}\label{eq:oldEveLO}
	UI(S;Y \backslash Z,Z') \leq UI(S;Y \backslash Z).
	\end{align}
	
	If $SY$--$Z$--$Z'$ is a Markov chain by assumption, 
	then
	\begin{align}\label{eq:newEveLO}  
              UI(S;Y\backslash Z,Z') & = \min_{Q'\in\Delta_{P'(S,Y,ZZ')}} I_{Q'}(S;Y|ZZ')
          \notag\\ &
          \stackrel{(a)}{=} \min_{Q'\in\Delta_{P'(S,Y,ZZ')}} I_{Q'}(S;Y|Z) 
          -I_{Q'}(S;Z'|Z)+I_{Q'}(S;Z'|ZY)\notag\\
	&\stackrel{(b)}{\ge}  \min_{Q'\in\Delta_{P'(S,Y,ZZ')}} I_{Q'}(S;Y|Z) \notag\\ 
	&\stackrel{(c)}{\ge} \min_{Q\in\Delta_{P(S,Y,Z)}} I_{Q}(S;Y|Z) = UI(S;Y\backslash Z), 
	\end{align}
	where~(a) follows from the chain rule of mutual information,
	(b) follows since $SY$--$Z$--$Z'$ w.r.t.~$P'$ implies $I_{Q'}(S;Z'|Z)=0$, and
	(c) follows since~$Q$ is the~$(S,Y,Z)$-marginal of~$Q'$ and~$Q'\in\Delta_{P'}$ implies~$Q\in\Delta_{P}$.
	\eqref{eq:oldEveLO} and~\eqref{eq:newEveLO} together imply $UI(S;Y\backslash Z)= UI(S;Y\backslash Z,Z')$. 
	
	Since \eqref{eq:oldEveLO} holds for any~$(S,Y,Z,Z')$, exchanging~$Z'$ and~$Z$ in~\eqref{eq:oldEveLO} gives $UI(S;Y\backslash Z)= UI(S;Y\backslash Z,Z')\le UI(S;Y\backslash Z')$ which completes the proof.
\end{proof}

\begin{proof} [Proof of Property~\ref{lem:UISKlock}]
	Let~$(S,Y,Z,U)\sim \widetilde{P}$ and let~$P$ be the $(S,Y,Z)$-marginal of~$\widetilde{P}$. Let
	\begin{align*} 
	\widetilde{Q}^{\ast} &\in \argmin_{\widetilde{Q}\in\Delta_{\widetilde{P}(S,Y,ZU)}} I_{\widetilde{Q}}(S;Y|ZU), \text{ and } 
	Q^{\ast} \in \argmin_{Q\in\Delta_{P(S,Y,Z)}} I_Q(S;Y|Z).
	\end{align*}
	Then
	\begin{align*} 
              UI(S;Y\backslash ZU) = I_{\widetilde{Q}^{\ast}}(S;Y|ZU) 
	\ge I_{\widetilde{Q}^{\ast}}(S;Y|Z)-H(U)  \ge I_{Q^{\ast}}(S;Y|Z)-H(U)
	=UI(S;Y\backslash Z)-H(U), 
	\end{align*}
	where in the third step we have used the fact that for any $\widetilde{Q}\in \Delta_{\widetilde{P}}$, the $(S,Y,Z)$-marginal of $\widetilde{Q}$ lies in~$\Delta_{P}$.
\end{proof}

\vspace{1cm}
\bibliographystyle{IEEEtran}
\bibliography{IEEEabrv,general}

\end{document}